\documentclass[12pt]{article}

\usepackage{amsmath,amssymb,amsfonts,amsthm}
\usepackage{algorithm}
\usepackage{algorithmic}
\usepackage{booktabs}
\usepackage{multirow}
\usepackage{array}
\usepackage{url}
\usepackage{xcolor}
\usepackage{graphicx}
\usepackage{float}
\usepackage{placeins}
\usepackage{geometry}
\usepackage{hyperref}
\usepackage{adjustbox}

\newtheorem{theorem}{Theorem}

\newtheorem{corollary}{Corollary}
\newtheorem{proposition}{Proposition}
\newtheorem{definition}{Definition}
\newtheorem{remark}{Remark}
\newtheorem{example}{Example}

\newcommand{\Zn}{\mathbb Z_n}

\newcommand{\Cay}{\operatorname{Cay}}
\newcommand{\RS}{\mathcal R}

\newcommand{\one}{\mathbf 1}

\title{Fault-Tolerant Shared-Relay Communication in Circulant Interconnection Networks}

\author{Bader Albader, Galal Hassan, and Mohamed R. Al-Mulla\\
\small Department of Computer Science, Faculty of Science, Kuwait University, Kuwait\\
\small \texttt{albader@cs.ku.edu.kw}}

\date{}

\begin{document}

\maketitle

\begin{abstract}
Circulant interconnection networks provide symmetric addressing, compact generator descriptions, and uniform local connectivity. This paper maps a degree--redundancy landscape for a fault-tolerant two-hop primitive in directed circulants: given $n$ nodes and degree budget $m$, how large can the worst-case shared-relay multiplicity $R(n,m)$ be? A node is a shared relay for an ordered terminal pair if it has outgoing links to both terminals; an $f$-relay-fault-tolerant circulant requires at least $f+1$ such relays for every pair. The underlying feasibility condition is a cyclic difference-multiplicity condition, which we use as a mathematical tool rather than claim as a new object. The contribution is the network-design framework around this tool: the parameters $R(n,m)$ and $D_f(n)$, a negative theorem for interval circulants, relay-table preprocessing and lookup algorithms, adversarial and random failure guarantees, load-balance scope, certified upper-bound interpretation of heuristic designs, exact small-$n$ calibration, a software lookup-versus-search microbenchmark, and a reproducible study of 526,539 generator sets. The results show that generator choice critically determines worst-case relay survivability: optimized threshold designs achieve $f$-relay-fault tolerance within about $1.16$--$1.63$ of the counting lower bound, while standard interval generators can fail structurally even at much larger degrees.
\end{abstract}


\section{Introduction}

Circulant interconnection networks are natural models for parallel and distributed systems because a node can be addressed by one integer modulo $n$ and its outgoing links are obtained by adding a fixed generator set. A directed circulant is written as
\[
G=\Cay(\Zn,S),
\]
where $S\subseteq\Zn$ and
\[
x\to y \quad \Longleftrightarrow \quad y-x\in S.
\]
Thus every node has the same out-degree $|S|$, the same local connection pattern, and the same routing table up to translation.

This paper studies a local fault-tolerance primitive that is different from conventional graph connectivity, diameter, or fault diameter. For two distinct terminal nodes $u,v\in\Zn$, a node $r$ is a \emph{shared relay} for the ordered pair $(u,v)$ if
\[
r\to u \quad\text{and}\quad r\to v.
\]
If every ordered terminal pair has many shared relays, then pairwise relay-assisted communication, acknowledgement aggregation, rendezvous communication, replicated control signaling, or local recovery can survive relay-node failures without running a global rerouting search.

Fig.~\ref{fig:small-circulants} illustrates the idea on two small directed circulants. The same terminal offset can have no shared relay under a poorly distributed generator set, while a more distributed generator set creates multiple immediate relay alternatives for the same pair.

\begin{figure}[H]
\centering
\includegraphics[width=0.75\linewidth]{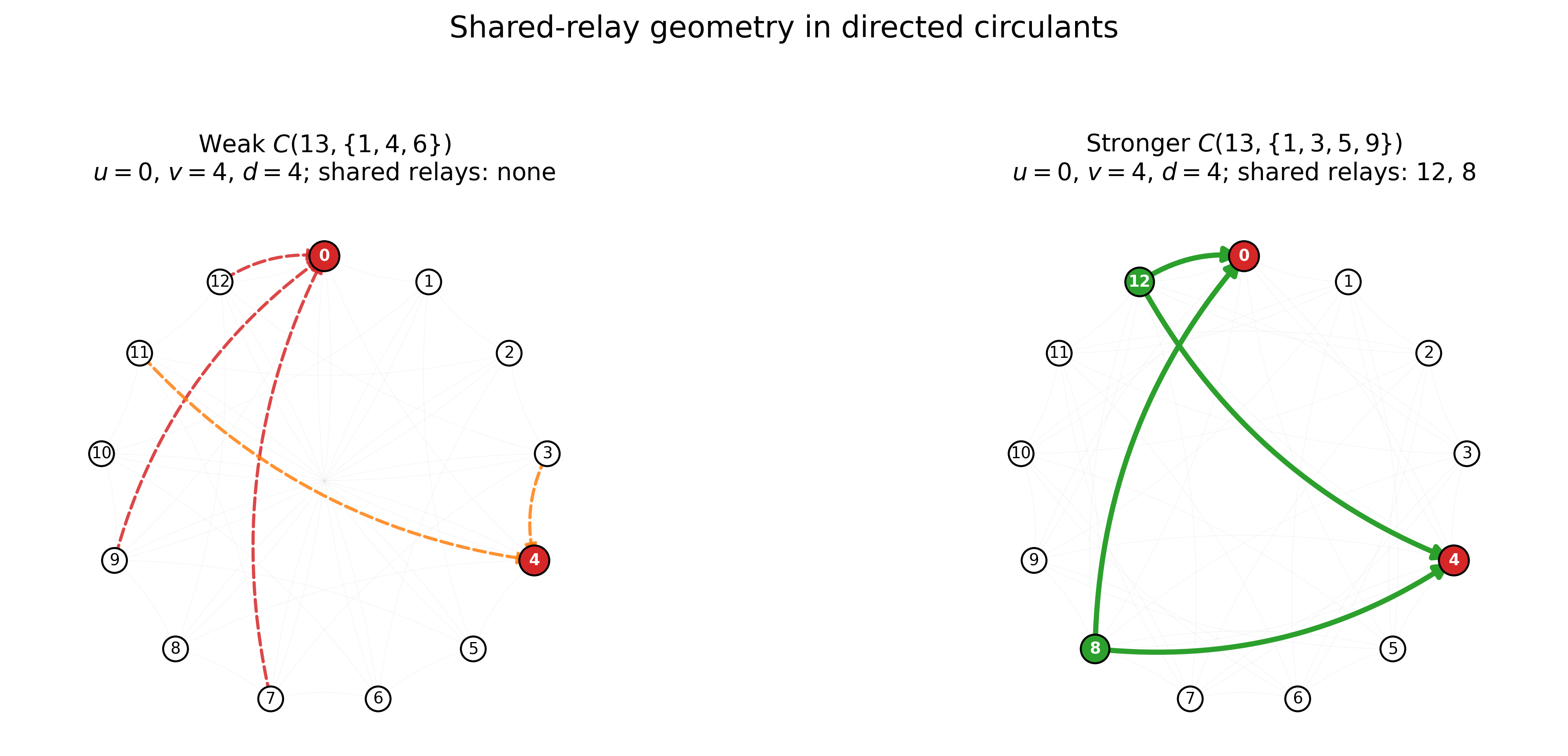}
\caption{Shared relays in two small directed circulants. Nodes are arranged on a ring. Red nodes are the terminals. Green nodes are shared relays and highlighted green arcs show the two outgoing relay links. For the same terminal offset, the weaker generator set has no shared relay. Red dashed arcs indicate one-sided attempts that reach $u$ but not $v$, while orange dashed arcs indicate attempts that reach $v$ but not $u$. The stronger generator set has two shared relays.}
\label{fig:small-circulants}
\end{figure}

The relevant design question is a degree--redundancy question: how much circulant degree is required to guarantee many immediate relay alternatives for every terminal pair? We formalize this through two parameters. The first is $D_f(n)$, the minimum directed circulant degree required so that every ordered terminal pair survives any $f$ relay failures. The second is $R(n,m)$, the maximum worst-case shared-relay multiplicity achievable by a degree-$m$ circulant on $n$ nodes.

The difference-multiplicity equivalence is the bridge, not the destination. The network-design question studied here is: for a fixed degree budget $m$ on $n$ nodes, what worst-case relay redundancy $R(n,m)$ is achievable, and how should a network exploit the resulting relay alternatives after failures? This $R(n,m)$ landscape has not previously been mapped for practical circulant generator families. The paper therefore treats known difference-basis theory as background and focuses on the network-layer consequences: failure semantics, relay tables, recovery algorithms, load metrics, memory costs, exact small-instance calibration, measured lookup costs, and empirical design guidance for circulant interconnection networks.

\subsection{Contributions}

The paper makes the following contributions, grouped by role.
\begin{itemize}
\item \textbf{Network metric and theory.} It defines shared-relay multiplicity as a local fault-tolerance metric, introduces $D_f(n)$ and $R(n,m)$ as degree--redundancy parameters, and proves the exact relay--difference equivalence, counting lower bounds, monotonicity, transfer bounds, adversarial recovery guarantees, and random-failure survival probabilities.
\item \textbf{Implementable relay layer.} It gives preprocessing, deterministic recovery, load-aware relay selection, and vectorized greedy construction algorithms. The recovery algorithms are proved correct, while the greedy construction is explicitly treated as a heuristic whose outputs are evaluated exactly.
\item \textbf{Systems cost and balance.} It quantifies preprocessing, routing-table storage, online lookup cost, and a scoped translation-invariant relay-balance property for complete offset sweeps, while explicitly stating that non-uniform traffic requires measured load counters or traffic-aware ordering.
\item \textbf{Reproducible design study.} It reports a checkpointed experiment over 526,539 circulant designs with $n\in\{251,503,1009,2003,5003,10007\}$ and $f\in\{0,\ldots,5\}$. The optimized run completed in about 3995 seconds and produced exact multiplicity values for every design.
\end{itemize}

\subsection{What This Enables Beyond the Difference-Basis Statement}

The difference-basis equivalence alone says that a set of generator differences covers every terminal offset. It does not specify how a network should store relay alternatives, skip failed relays, choose among multiple surviving relays, quantify relay-table memory, compare adversarial and random failure models, or expose why common generator families behave poorly under worst-case relay redundancy. Those questions are the systems contribution of this paper. The framework converts the abstract multiplicity condition into a routing-table primitive: for every offset $d$, store the offset pairs $(a,b)$ with $b-a=d$; for terminals $(u,v)$, convert each $(a,b)$ into relay $r=u-a=v-b$; skip failed relays; optionally choose the least-loaded survivor.

\section{Related Work}
\label{sec:related}

\subsection{Circulant and Cayley Interconnection Networks}

Circulant and Cayley networks have long been studied because they combine symmetry, compact addressing, and efficient routing. Their standard matrix representation and algebraic shift structure are treated in classical circulant-matrix references \cite{GrayToeplitzCirculant2005,KalmanWhite2001}; this is the algebraic basis for representing every local connection pattern by a single generator vector. Classical network work considers connectivity, diameter, broadcasting, containers, wide diameter, fault diameter, and routing in circulant and Cayley networks \cite{BoeschTindell1984,BermondComellasHsu1995,Hwang2003,Monakhova2012,LiawChangTang1998,Xu2001,Duato2003}. Circulant and loop-network routing also appears in fault-tolerant distributed-loop routing and recent network-on-chip studies \cite{MukhopadhyayaSinha1995,RomanovAccess2020,MonakhovaTNSE2023}. Ring interconnects, which are degree-two circulants, have also appeared in commercial and research multicore settings \cite{IntelRingInterconnect,ArimilliPERCS2010}. The present paper addresses a complementary local primitive: common in-neighborhood relay multiplicity for immediate two-hop relay alternatives without path search.

\subsection{Difference Bases, $g$-Difference Bases, and Difference Covers}

A subset $B$ of an abelian group $G$ is a difference basis if every group element can be represented as $a-b$ for some $a,b\in B$. Banakh and Gavrylkiv studied difference bases in cyclic groups and proved lower and upper bounds for the cyclic difference size \cite{BanakhGavrylkiv2019}. Multiplicity variants require every element to have at least $g$ difference representations. Schmutz and Tait studied cardinalities of $g$-difference bases and proved asymptotic results for normalized sizes in their setting \cite{SchmutzTait2025}. Recent work of Li and Yip treats generalized additive and difference bases in finite abelian groups \cite{LiYip2025}. The latter two sources are cited as arXiv preprints; therefore, the multiplicity-$g$ asymptotic discussion below is stated as a transfer statement rather than as a new unconditional theorem for cyclic circulants.

\subsection{Common-Neighbor Counts}

Common-neighbor counts also arise in special Cayley and circulant graphs. Klotz and Sander derived formulas for common neighbors in unitary Cayley graphs \cite{KlotzSander2007}. In contrast, the present paper attaches a fault-tolerance semantics to the common in-neighborhood: it is the set of immediate relay alternatives that survive relay-node failures.

\section{Network Model, Notation, and Parameters}
\label{sec:model}

\begin{table}[H]
\centering
\caption{Notation used throughout the paper.}
\label{tab:notation}
\begin{tabular}{ll}
\toprule
Symbol & Meaning \\
\midrule
$n$ & number of nodes \\
$\Zn$ & cyclic node set $\{0,\ldots,n-1\}$ \\
$S$ & circulant generator set, $S\subseteq\Zn\setminus\{0\}$ \\
$m$ & directed degree, $m=|S|$ \\
$G$ & directed circulant $\Cay(\Zn,S)$ \\
$u,v$ & ordered terminal pair, $u\ne v$ \\
$d$ & terminal offset $d=v-u\pmod n$ \\
$\RS(u,v)$ & shared-relay set of $(u,v)$ \\
$\lambda_S(d)$ & number of ordered generator pairs with $b-a=d$ \\
$R(S)$ & worst-case nonzero relay multiplicity \\
$R(n,m)$ & best possible $R(S)$ among degree-$m$ circulants \\
$D_f(n)$ & minimum degree for $f$-relay-fault tolerance \\
$F$ & failed relay-node set \\
$P_d$ & relay-offset table for terminal offset $d$ \\
\bottomrule
\end{tabular}
\end{table}

Let $n\ge2$ and let $S\subseteq\Zn$. The directed circulant $G=\Cay(\Zn,S)$ has vertex set $\Zn$ and directed edge set
\[
E=\{(x,y):y-x\in S\}.
\]
The directed degree is $m=|S|$. Unless otherwise stated, $0\notin S$, so self-loops are excluded. For a node $u$, its in-neighborhood is
\[
N^-(u)=\{r\in\Zn:u-r\in S\}=u-S.
\]

\begin{definition}[Shared-relay set]
For distinct $u,v\in\Zn$, the shared-relay set of $(u,v)$ is
\[
\RS(u,v)=N^-(u)\cap N^-(v).
\]
Thus $r\in\RS(u,v)$ if and only if $r\to u$ and $r\to v$.
\end{definition}

\begin{definition}[Difference multiplicity]
For $d\in\Zn$, define
\[
\lambda_S(d)=|\{(a,b)\in S^2:b-a=d\}|.
\]
The worst-case nonzero shared-relay multiplicity is
\[
R(S)=\min_{d\in\Zn\setminus\{0\}}\lambda_S(d).
\]
\end{definition}

The definition of $\lambda_S(d)$ includes all ordered pairs in $S^2$. For $d\ne0$, every counted pair automatically has $a\ne b$. For $d=0$, because $S$ is a set, the only counted pairs are $(a,a)$ and hence $\lambda_S(0)=m$. Throughout the paper, fault tolerance depends only on nonzero offsets.

\begin{definition}[$f$-relay-fault tolerance]
The circulant $\Cay(\Zn,S)$ is $f$-relay-fault-tolerant if for every ordered pair $u\ne v$ and every set $F\subseteq\Zn$ with $|F|\le f$, there exists a relay
\[
r\in\RS(u,v)\setminus F.
\]
\end{definition}

\begin{definition}[Degree--redundancy parameters]
\label{def:Rnm}
For $0\le m\le n-1$, define
\[
R(n,m)=\max_{\substack{S\subseteq\Zn\setminus\{0\}\\ |S|=m}} R(S).
\]
For the base case $m=0$, the maximum is taken over the single empty generator set and $R(n,0)=0$. For $f\ge0$, define
\[
D_f(n)=\min\{m\in\{0,\ldots,n-1\}:R(n,m)\ge f+1\},
\]
when such a degree exists.
\end{definition}

\section{Relay--Difference Equivalence and Design Bounds}
\label{sec:equiv}

\begin{theorem}[Relay--difference equivalence]
\label{thm:equiv}
For any distinct $u,v\in\Zn$,
\[
|\RS(u,v)|=\lambda_S(v-u).
\]
Moreover, the map
\[
(a,b)\in S^2,\quad b-a=v-u
\]
to
\[
r=u-a=v-b
\]
is a bijection between difference representations of $v-u$ and shared relays of $(u,v)$.
\end{theorem}

\begin{proof}
If $(a,b)\in S^2$ and $b-a=v-u$, define $r=u-a$. Then $u-r=a\in S$ and
\[
v-r=v-u+a=b\in S,
\]
so $r\in\RS(u,v)$. Conversely, if $r\in\RS(u,v)$, set $a=u-r$ and $b=v-r$. Then $a,b\in S$ and
\[
b-a=(v-r)-(u-r)=v-u.
\]
The two maps are inverse to each other, proving the bijection. Counting gives the equality.
\end{proof}

\begin{corollary}
\label{cor:fiff}
A directed circulant $\Cay(\Zn,S)$ is $f$-relay-fault-tolerant if and only if
\[
R(S)\ge f+1.
\]
\end{corollary}

\begin{proof}
By Theorem~\ref{thm:equiv}, every ordered pair with offset $d\ne0$ has exactly $\lambda_S(d)$ shared relays. Thus every pair has at least $f+1$ shared relays if and only if $R(S)\ge f+1$. Removing at most $f$ failed relays cannot exhaust a set of size at least $f+1$. Conversely, if some pair $(u,v)$ has at most $f$ relays, the adversary may choose $F=\RS(u,v)$; then $|F|\le f$ and every shared-relay option for that pair is failed.
\end{proof}

\begin{theorem}[Counting upper bound on $R(n,m)$]
\label{thm:count-upper}
For every $n$ and $m$,
\[
R(n,m)\le \left\lfloor\frac{m(m-1)}{n-1}\right\rfloor.
\]
\end{theorem}

\begin{proof}
For nonzero offsets, every ordered pair $(a,b)\in S^2$ with $a\ne b$ contributes to exactly one $d=b-a\ne0$. Since there are $m(m-1)$ ordered pairs with $a\ne b$,
\[
\sum_{d\ne0}\lambda_S(d)=m(m-1).
\]
The minimum over $n-1$ nonzero offsets is at most the average. Maximizing over $S$ gives the result.
\end{proof}

\begin{theorem}[Interval circulants are structurally weak]
\label{thm:interval-negative}
Let $S_m=\{1,2,\ldots,m\}\subseteq\mathbb Z_n$. If $n>2m-1$, then
\[
R(S_m)=0.
\]
More precisely, for $1\le d\le m-1$, $\lambda_{S_m}(d)=m-d$; for $n-m+1\le d\le n-1$, $\lambda_{S_m}(d)=m-(n-d)$; and all remaining nonzero offsets have multiplicity zero.
\end{theorem}

\begin{proof}
For $a,b\in S_m$, the ordinary integer difference $b-a$ lies in $[-m+1,m-1]$. Therefore, modulo $n$, the only nonzero residues represented by such differences are $1,\ldots,m-1$ and $n-m+1,\ldots,n-1$. If $n>2m-1$, then $n-m+1>m$, so the upper arc $\{n-m+1,\ldots,n-1\}$ begins strictly above the lower arc $\{1,\ldots,m-1\}$; the two arcs are therefore disjoint and together cover exactly $2(m-1)$ of the $n-1$ nonzero offsets, leaving $n-1-2(m-1)=n-2m+1\ge1$ offsets with multiplicity zero. For $1\le d\le m-1$, the pairs are $(1,1+d),\ldots,(m-d,m)$, giving $m-d$ representations. The negative offsets are symmetric: residue $n-d$ has the $m-d$ representations $(1+d,1),\ldots,(m,m-d)$. All other nonzero offsets have no representation. The boundary is tight: when $n=2m-1$, the lower arc and upper arc partition all nonzero residues.
\end{proof}

This theorem gives a non-empirical reason why interval generator sets are poor shared-relay designs: unless their degree is already at least roughly half the network size, they are not even $0$-relay-fault-tolerant.

\begin{corollary}[Degree lower bound]
\label{cor:degree-lb}
If $\Cay(\Zn,S)$ is $f$-relay-fault-tolerant with $|S|=m$, then
\[
m(m-1)\ge (f+1)(n-1).
\]
Consequently,
\[
D_f(n)\ge
\left\lceil\frac{1+\sqrt{1+4(f+1)(n-1)}}{2}\right\rceil.
\]
\end{corollary}

\begin{proof}
By Corollary~\ref{cor:fiff}, every nonzero offset must have multiplicity at least $f+1$. Summing over the $n-1$ nonzero offsets gives the first inequality. The displayed lower bound follows by solving the quadratic inequality in $m$.
\end{proof}

\begin{proposition}[Monotonicity]
\label{prop:mono}
For fixed $n$, $R(n,m)$ is nondecreasing in $m$, and $D_f(n)$ is nondecreasing in $f$.
\end{proposition}

\begin{proof}
If $S$ has degree $m$ and $x\notin S\cup\{0\}$, then every old difference representation remains present in $S\cup\{x\}$. Hence $R(S\cup\{x\})\ge R(S)$. For every degree-$m$ set $S$, the set $S\cup\{x\}$ is one admissible degree-$(m+1)$ set, so
\[
R(n,m+1)=\max_{|T|=m+1} R(T)
   \ge R(S\cup\{x\})\ge R(S).
\]
Taking the maximum over all degree-$m$ sets $S$ gives $R(n,m+1)\ge R(n,m)$. Since the requirement $R(S)\ge f+1$ becomes stronger as $f$ increases, $D_f(n)$ is nondecreasing in $f$.
\end{proof}

\begin{proposition}[Verifiable certificates for relay redundancy]
\label{prop:certificates}
For a fixed generator set $S$ with $|S|=m$, the value $R(S)$ can be verified in $O(m^2+n)$ time and $O(n)$ auxiliary space by enumerating all ordered pairs $(a,b)\in S^2$, incrementing the counter for $d=b-a$, and taking the minimum over the $n-1$ nonzero counters. Consequently, the decision problem
\[
\exists S\subseteq \Zn\setminus\{0\},\quad |S|=m,
\quad R(S)\ge t
\]
belongs to NP when $S$ is used as the certificate.
\end{proposition}

\begin{proof}
There are $m^2$ ordered pairs in $S^2$. Each pair contributes to exactly one modular difference counter, and the $m$ diagonal pairs contribute only to the zero offset. After all counters are built, scanning the $n-1$ nonzero counters gives $R(S)$. If $R(S)\ge t$, the set $S$ is a polynomial-size certificate for the displayed existential statement, and the verification just described is polynomial in the input size when the group elements are listed explicitly.
\end{proof}

\begin{corollary}[Certified upper bounds from constructed designs]
\label{cor:certified-upper}
If an experiment outputs a degree-$m$ generator set $S$ with $R(S)\ge f+1$, then the output is a verifiable certificate that
\[
D_f(n)\le m.
\]
If an exhaustive search additionally proves that no degree below $m$ achieves $R(S)\ge f+1$, then $D_f(n)=m$.
\end{corollary}

\begin{proof}
The first statement is immediate from Definition~\ref{def:Rnm} and Corollary~\ref{cor:fiff}. The second adds the matching nonexistence certificate over all smaller degrees.
\end{proof}

\begin{proposition}[Transfer from loop-tolerant multiplicity]
\label{prop:zero-removal}
Let $A\subseteq\Zn$ be a set, possibly containing $0$, such that every nonzero $d$ has at least $g$ ordered representations $b-a=d$ with $a,b\in A$. Let $S=A\setminus\{0\}$. Then for every $d\ne0$,
\[
\lambda_S(d)\ge \lambda_A(d)-2.
\]
Consequently, if $g\ge f+3$, then $S$ is $f$-relay-fault-tolerant.
\end{proposition}

\begin{proof}
For a fixed nonzero $d$, the representations counted by $\lambda_A(d)$ that use $0$ are only $(0,d)$ and $(-d,0)$, when the corresponding nonzero elements belong to $A$. Thus deleting $0$ removes at most two representations. The final statement follows from $\lambda_S(d)\ge g-2\ge f+1$.
\end{proof}

\begin{example}[The zero-removal loss can be two]
In $\mathbb Z_8$, let $A=\{0,1,6,7\}$ and consider $d=1$. The representations are
\[
(0,1),\quad (6,7),\quad (7,0),
\]
so $\lambda_A(1)=3$. Removing $0$ leaves $S=\{1,6,7\}$ and only $(6,7)$ remains, so $\lambda_S(1)=1$. The loss is exactly two. By contrast, if $A=\{0,1,6\}$, then deleting $0$ removes only the representation $(0,1)$ for $d=1$.
\end{example}

\begin{theorem}[Unconditional lower-bound scale]
\label{thm:asymptotic}
For every fixed $f\ge0$,
\[
D_f(n)=\Omega\!\left(\sqrt{(f+1)n}\right).
\]
More explicitly,
\[
D_f(n)\ge
\left\lceil\frac{1+\sqrt{1+4(f+1)(n-1)}}{2}\right\rceil.
\]
\end{theorem}

\begin{proof}
This is exactly Corollary~\ref{cor:degree-lb}; the asymptotic form follows because the leading term of the displayed lower bound is $\sqrt{(f+1)n}$.
\end{proof}

\begin{remark}[Conditional construction-transfer context]
For $f=0$, the lower-bound scale can be compared with peer-reviewed cyclic difference-basis upper bounds such as those of Banakh and Gavrylkiv~\cite{BanakhGavrylkiv2019}. For $f\ge1$, we do not state a peer-reviewed unconditional cyclic upper-bound theorem here. Instead, any explicit cyclic multiplicity-$(f+3)$ difference-basis construction of size $O_f(\sqrt n)$ would transfer through Proposition~\ref{prop:zero-removal} to a directed circulant family with $D_f(n)=O_f(\sqrt n)$. This observation is used only as context for future hybrid design and not as a claimed theorem of the present paper.
\end{remark}

\subsection{Worked Difference-Multiplicity Examples}

\begin{example}[A small generator set]
Let $n=13$ and $S=\{1,3,9\}$. The nonzero multiplicities are
\[
\begin{array}{c|cccccccccccc}
d&1&2&3&4&5&6&7&8&9&10&11&12\\\hline
\lambda_S(d)&0&1&0&0&1&1&1&1&0&0&1&0.
\end{array}
\]
Therefore $R(S)=0$. For $d=4$, there is no pair $(a,b)\in S^2$ with $b-a\equiv4\pmod{13}$, so a terminal pair at offset $4$ has no shared relay.
\end{example}

\begin{example}[Interval generators fail worst-case coverage]
Let $n=13$ and $S=\{1,2,3,4\}$. Then
\[
(\lambda_S(1),\ldots,\lambda_S(12))=(3,2,1,0,0,0,0,0,0,1,2,3),
\]
so $R(S)=0$. In particular, offsets $4,5,6,7,8,9$ have no shared relays. A same-degree translated Singer difference set, for example $S'=\{1,2,4,10\}$, has $\lambda_{S'}(d)=1$ for every $d\ne0$ in $\mathbb Z_{13}$ and hence is $0$-relay-fault-tolerant. This is self-contained in the following ordered-pair table:
\begin{center}
\scriptsize
\begin{tabular}{c c|c c|c c}
\toprule
$d$ & $(a,b)$ & $d$ & $(a,b)$ & $d$ & $(a,b)$ \\
\midrule
1&(1,2)&2&(2,4)&3&(1,4)\\
4&(10,1)&5&(10,2)&6&(4,10)\\
7&(10,4)&8&(2,10)&9&(1,10)\\
10&(4,1)&11&(4,2)&12&(2,1)\\
\bottomrule
\end{tabular}
\end{center}
Each entry satisfies $b-a\equiv d\pmod{13}$; for example, the $d=7$ entry uses $(a,b)=(10,4)$ because $4-10\equiv -6\equiv7\pmod{13}$.
\end{example}

Fig.~\ref{fig:spectrum} visualizes this phenomenon at a larger size. The worst-case bar, not the average bar, determines the relay-fault tolerance.

\begin{figure}[H]
\centering
\includegraphics[width=0.75\linewidth]{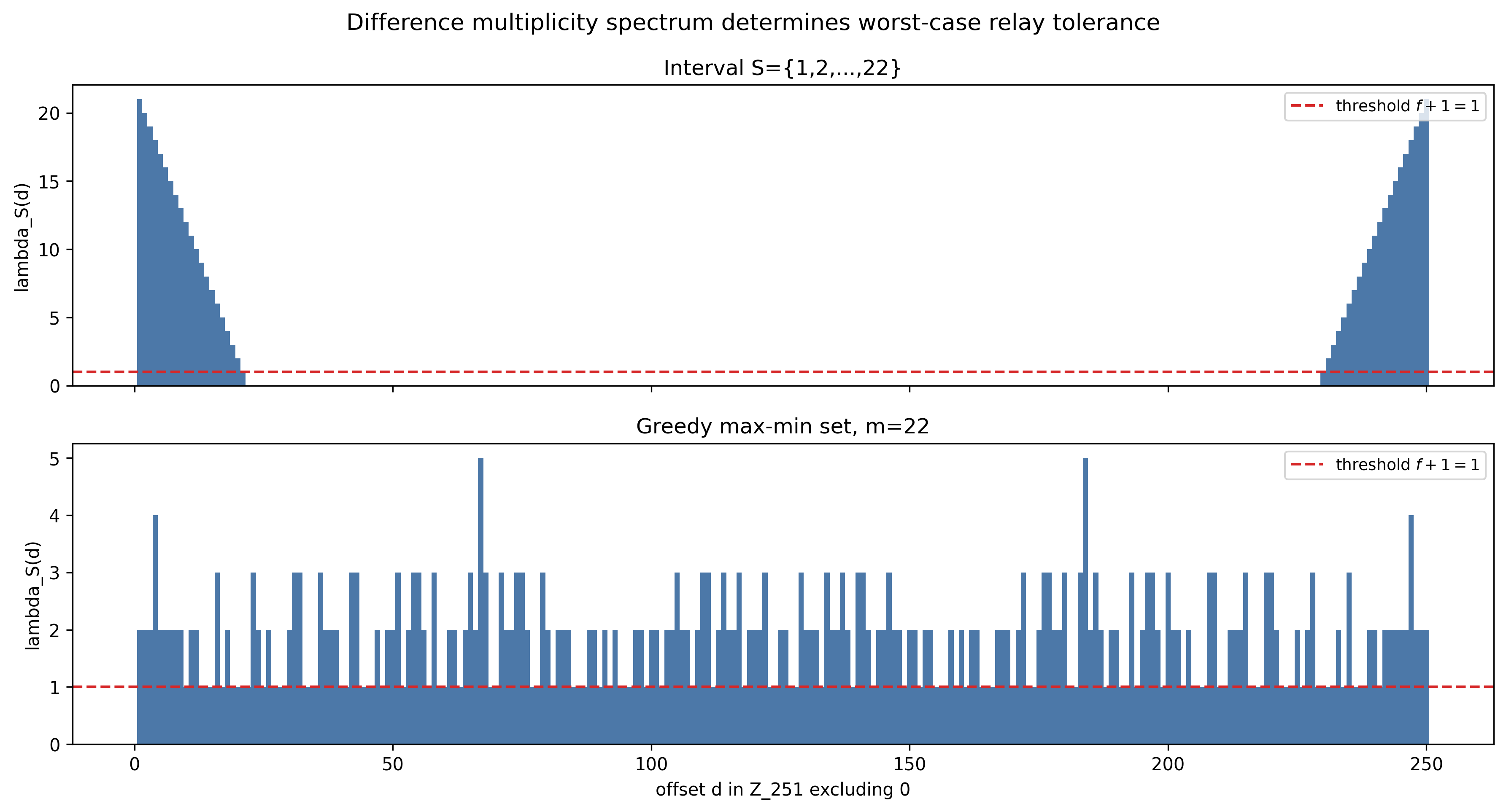}
\caption{Difference-multiplicity spectra for two degree-$22$ generator sets on $n=251$. The interval set has many offsets below the threshold, while the greedy set lifts all nonzero offsets to at least one shared relay.}
\label{fig:spectrum}
\end{figure}

\section{Failure Models and Recovery Guarantees}
\label{sec:failure}

\begin{theorem}[Worst-case relay recovery]
\label{thm:worst}
If $R(S)\ge f+1$, then after any failed relay set $F\subseteq\Zn$ with $|F|\le f$, every ordered terminal pair $u\ne v$ has a surviving relay in $\RS(u,v)\setminus F$.
\end{theorem}

\begin{proof}
By Theorem~\ref{thm:equiv}, $|\RS(u,v)|=\lambda_S(v-u)\ge R(S)\ge f+1$. A set $F$ of at most $f$ failed nodes cannot contain all elements of $\RS(u,v)$.
\end{proof}

\begin{theorem}[Random relay failures]
\label{thm:random}
Fix a terminal pair $(u,v)$ with $k=|\RS(u,v)|$. If $0\le q\le n$ failed relay nodes are chosen uniformly at random from $\Zn$, then the pair survives with probability
\[
P_{\mathrm{surv}}(n,k,q)=
1-\frac{\binom{n-k}{q-k}}{\binom{n}{q}},
\]
where the numerator is interpreted as $0$ if $q<k$.
\end{theorem}

\begin{proof}
The pair fails exactly when all $k$ of its shared relays are included among the $q$ failed nodes. If $q<k$, this event is impossible. If $q\ge k$, the number of failure sets containing all $k$ relays is $\binom{n-k}{q-k}$, while the number of all $q$-node failure sets is $\binom{n}{q}$. Subtracting this failure probability from one gives the formula.
\end{proof}

\begin{corollary}[Independent relay failures]
If each relay node fails independently with probability $p$, then a terminal pair with $k$ shared relays fails with probability $p^k$. In particular, if $R(S)\ge f+1$, then every pair has independent-failure survival probability at least $1-p^{f+1}$.
\end{corollary}

\begin{proof}
A pair fails exactly when all of its $k$ shared relays fail. Under independent node failures this event has probability $p^k$. Since $k\ge R(S)\ge f+1$, the failure probability is at most $p^{f+1}$.
\end{proof}

\begin{example}[Worst-case versus random failures]
For a worst-case pair with $k=R(S)=6$ relays in a network with $n=251$, any adversarial set of at most five failed relays leaves at least one relay alive. Under random failures, the same pair has survival probability exactly one for $q\le5$ because fewer than six failures cannot kill all six relays. For $q=10$, Theorem~\ref{thm:random} gives
\[
1-\frac{\binom{245}{4}}{\binom{251}{10}}\approx 0.99999996.
\]
The adversarial threshold is therefore conservative but essential: it gives a deterministic guarantee, while random failures are typically much less damaging.
\end{example}

Fig.~\ref{fig:random-vs-worst} compares the exact random-failure formula for a worst-case pair against empirical mean survival over sampled terminal pairs.

\begin{figure}[H]
\centering
\includegraphics[width=0.6\linewidth]{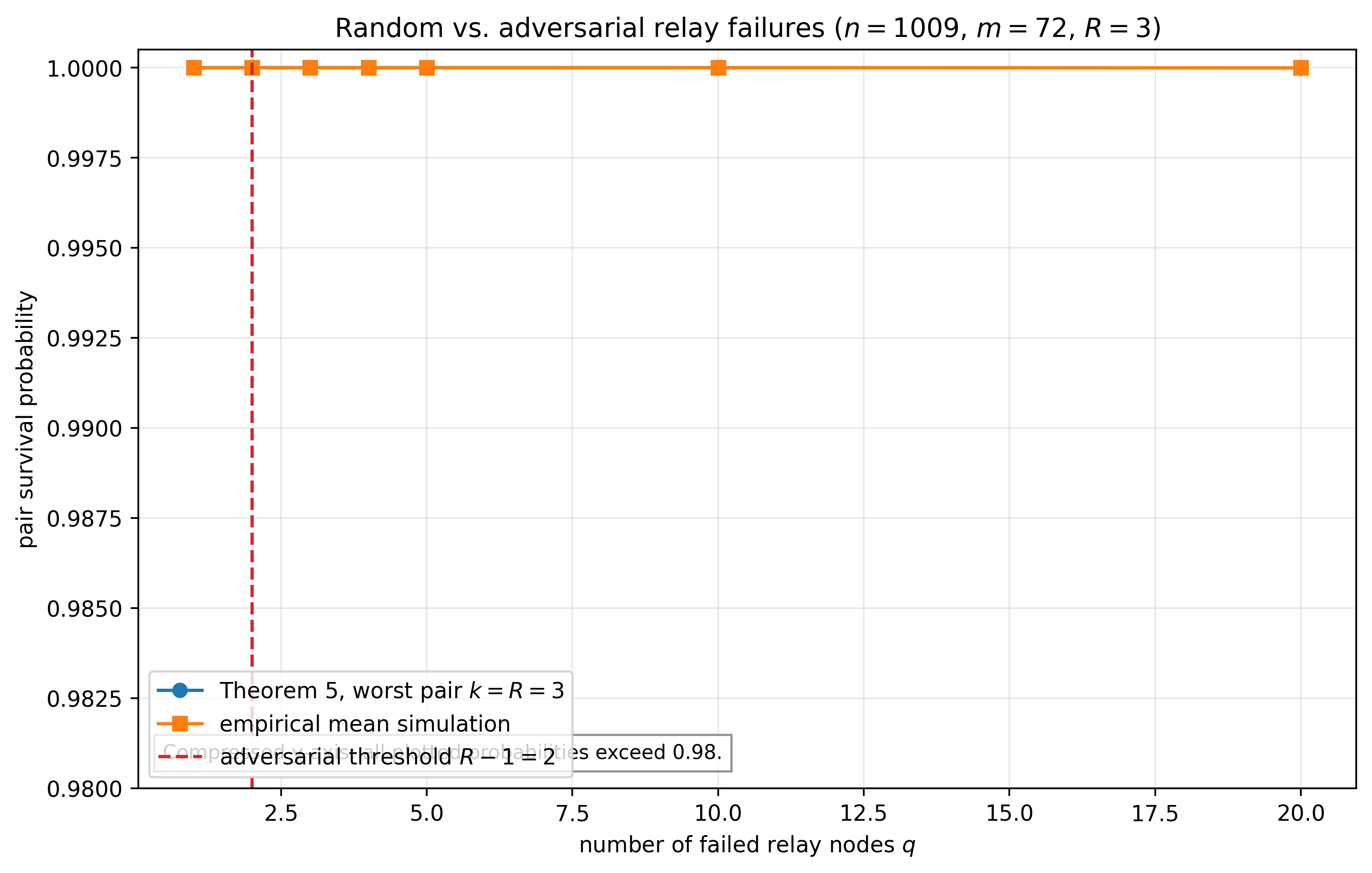}
\caption{Random versus adversarial relay failures for a representative $n=1009$ greedy design. The vertical line marks the deterministic adversarial threshold $R(S)-1$. Random-failure survival remains high beyond this threshold because random failures are unlikely to target exactly all relays of a worst-case pair. The y-axis is intentionally compressed to distinguish probabilities that are all above $0.98$.}
\label{fig:random-vs-worst}
\end{figure}

\section{Relay Tables and Selection Algorithms}
\label{sec:algorithms}

For each nonzero offset $d$, define the relay-offset table
\[
P_d=\{(a,b)\in S^2:b-a=d\}.
\]
By Theorem~\ref{thm:equiv}, if $d=v-u$ and $(a,b)\in P_d$, then the corresponding shared relay is
\[
r=u-a=v-b.
\]
Thus a single table indexed by offsets serves every terminal pair by translation.

\begin{proposition}[Table size]
\label{prop:table-size}
The total number of stored relay-offset entries over all nonzero offsets is
\[
\sum_{d\ne0}|P_d|=m(m-1).
\]
\end{proposition}

\begin{proof}
Each ordered pair $(a,b)\in S^2$ with $a\ne b$ contributes to exactly one nonzero offset $d=b-a$. There are $m(m-1)$ such ordered pairs.
\end{proof}

\begin{corollary}[Relay-table preprocessing time]
\label{cor:preprocess}
All tables $P_d$ can be constructed in $O(m^2)$ time by enumerating the ordered pairs $(a,b)\in S^2$ and inserting each pair with $a\ne b$ into the table indexed by $d=b-a$.
\end{corollary}

\begin{algorithm}[H]
\caption{Deterministic surviving-relay lookup}
\label{alg:lookup}
\begin{algorithmic}[1]
\REQUIRE $n,S$, precomputed tables $P_d$, terminals $u\ne v$, failed-node bitmap $F$
\STATE $d\leftarrow (v-u)\bmod n$
\FORALL{$(a,b)\in P_d$ in deterministic order}
    \STATE $r\leftarrow (u-a)\bmod n$
    \IF{$F[r]=0$}
        \RETURN $r$
    \ENDIF
\ENDFOR
\RETURN failure
\end{algorithmic}
\end{algorithm}

\begin{theorem}[Correctness of Algorithm~\ref{alg:lookup}]
If $|F|\le f$ and $R(S)\ge f+1$, then Algorithm~\ref{alg:lookup} returns a valid surviving relay for every $u\ne v$.
\end{theorem}

\begin{proof}
For $d=v-u$, the table $P_d$ contains exactly the difference representations corresponding to shared relays of $(u,v)$ by Theorem~\ref{thm:equiv}. Since $R(S)\ge f+1$, $|P_d|\ge f+1$. At most $f$ of the corresponding relays can belong to $F$, so at least one table entry produces a relay $r\notin F$. When Algorithm~\ref{alg:lookup} returns such an $r$, the associated $(a,b)$ satisfies $u-r=a\in S$ and $v-r=b\in S$, so $r\to u$ and $r\to v$.
\end{proof}

\begin{proposition}[Lookup cost versus graph search]
\label{prop:lookup-cost}
For terminal offset $d=v-u$, Algorithm~\ref{alg:lookup} inspects at most $\lambda_S(d)$ table entries. Hence its worst-case online scan length is at most
\[
\max_{d\ne0}\lambda_S(d)\le m,
\]
and, under a traffic model in which the terminal-pair offset $d$ is uniformly distributed over the nonzero offsets, its average scan length is
\[
\frac{1}{n-1}\sum_{d\ne0}\lambda_S(d)=\frac{m(m-1)}{n-1}.
\]
By contrast, a standard breadth-first rerouting search in a directed degree-$m$ circulant can inspect $\Theta(nm)$ edges in the worst case.
\end{proposition}

\begin{proof}
For a fixed offset $d$, the algorithm scans the entries of $P_d$, whose cardinality is $\lambda_S(d)$. For each $a\in S$ there is at most one $b=a+d$ in $S$, so $\lambda_S(d)\le m$. The average identity is Proposition~\ref{prop:table-size} divided by $n-1$ under the stated uniform-offset assumption.
\end{proof}

\begin{algorithm}[H]
\caption{Load-aware surviving-relay selection}
\label{alg:load}
\begin{algorithmic}[1]
\REQUIRE $n,S$, tables $P_d$, terminals $u\ne v$, failed bitmap $F$, relay load vector $L$
\STATE $d\leftarrow (v-u)\bmod n$; $best\leftarrow\bot$
\FORALL{$(a,b)\in P_d$}
    \STATE $r\leftarrow (u-a)\bmod n$
    \IF{$F[r]=0$ and ($best=\bot$ or $L[r]<L[best]$)}
        \STATE $best\leftarrow r$
    \ENDIF
\ENDFOR
\RETURN $best$
\end{algorithmic}
\end{algorithm}

\begin{theorem}[Load-aware correctness]
Under the same hypothesis as Theorem~\ref{thm:worst}, Algorithm~\ref{alg:load} returns a valid surviving relay. Among surviving relays listed in $P_{v-u}$, it returns one with minimum recorded load.
\end{theorem}

\begin{proof}
The existence and validity of at least one surviving relay follow as in the proof of Algorithm~\ref{alg:lookup}. Algorithm~\ref{alg:load} scans all surviving relays represented by $P_{v-u}$ and updates $best$ only when a strictly smaller load is found. Therefore the returned relay has minimum load among surviving candidates.
\end{proof}

\begin{proposition}[Translation-invariant relay balance]
\label{prop:round-robin-load}
Fix a nonzero offset $d$ and write $P_d=\{(a_1,b_1),\ldots,(a_k,b_k)\}$. For any fixed representation $j$, if each of the $n$ ordered terminal pairs $(u,u+d)$, $u\in\Zn$, issues exactly one relay request and all requests use representation $j$, then every relay node is selected exactly once. More generally, suppose a policy assigns, for each representation $j$, exactly $t_j$ complete translation sweeps over all $n$ ordered terminal pairs with offset $d$. Then every relay node is selected exactly $\sum_{j=1}^k t_j$ times.
\begin{proof}
For representation $(a_j,b_j)$ and terminal $u$, the selected relay is $r=u-a_j$. As $u$ ranges over $\Zn$, the map $u\mapsto u-a_j$ is a bijection of $\Zn$. Hence one full sweep over the $n$ terminal pairs with representation $j$ selects each relay exactly once. Repeating this sweep $t_j$ times contributes exactly $t_j$ selections to every relay. Summing over the representations gives exactly $\sum_{j=1}^k t_j$ selections at every relay node.
\end{proof}
\end{proposition}

The hypothesis of Proposition~\ref{prop:round-robin-load} is intentionally scoped to uniform complete sweeps over a fixed offset. Non-uniform traffic, arbitrary streaming request order, or repeated demand for a small subset of terminal pairs can concentrate load on specific relays; in that case Algorithm~\ref{alg:load} should be combined with measured load counters or traffic-aware ordering of $P_d$ rather than relying on symmetry alone.

\begin{example}[Algorithm trace]
Let $n=13$, $S=\{1,4,6,9\}$, $u=0$, and $v=5$. Then $d=5$ and
\[
P_5=\{(1,6),(4,9),(9,1)\}.
\]
The corresponding relays are $12$, $9$, and $4$. If $F=\{12\}$, Algorithm~\ref{alg:lookup} skips $12$ and may return $9$. Indeed,
\[
0-9\equiv4\in S,\qquad 5-9\equiv9\in S,
\]
so $9\to0$ and $9\to5$.
\end{example}

\begin{figure}[H]
\centering
\includegraphics[width=0.75\linewidth]{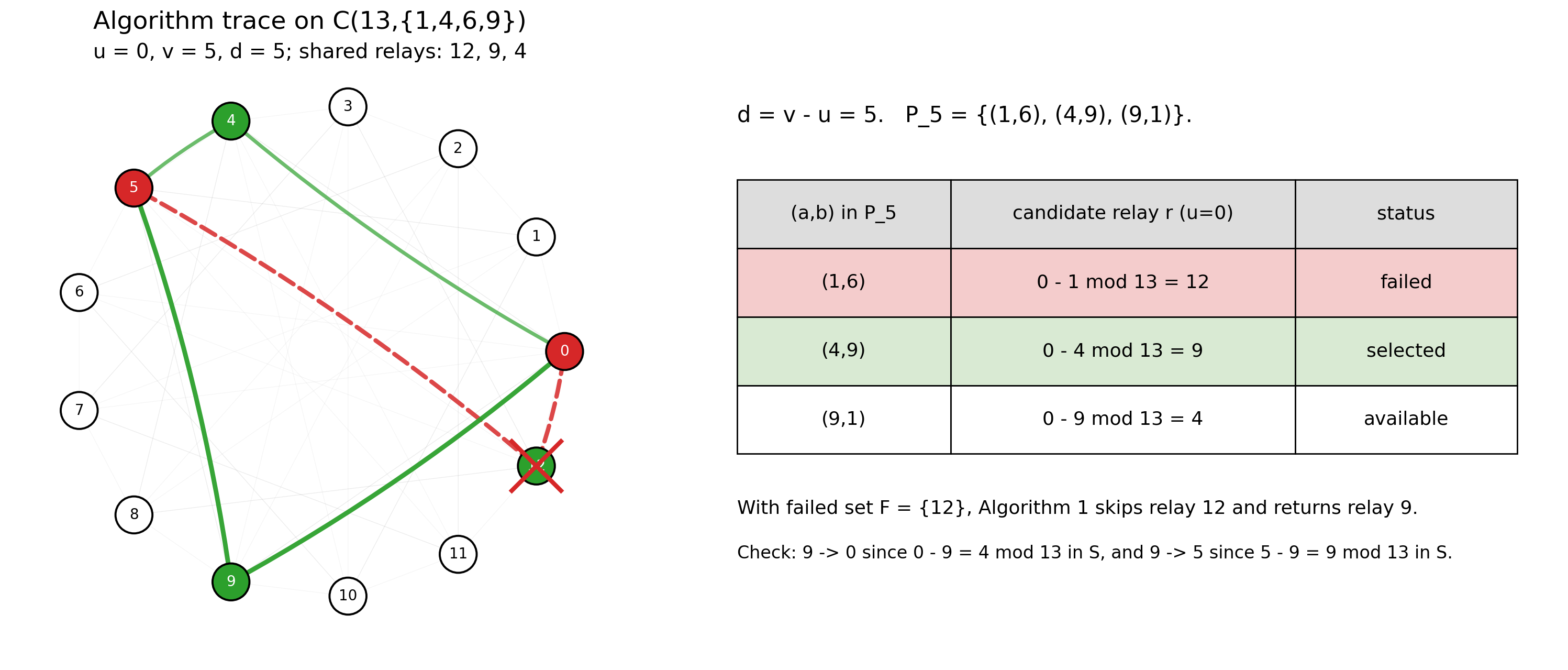}
\caption{Worked trace of Algorithm~\ref{alg:lookup}. The table $P_5$ produces three candidate relays for $u=0$; the arithmetic in the relay column substitutes $u=0$ explicitly. The failed relay is crossed out and the first surviving relay is selected.}
\label{fig:alg-trace}
\end{figure}

\section{Vectorized Greedy Generator Construction}
\label{sec:greedy}

The experiments use several generator families. The strongest threshold designs were produced by a vectorized greedy heuristic. The heuristic is not a proof of optimality; after each set is constructed, all multiplicities are recomputed exactly by the definition of $\lambda_S$.

Suppose the current partial generator set is $S$ with multiplicities $\lambda(d)$. If a new candidate $x\notin S\cup\{0\}$ is added, then for every nonzero $d$ the increment is
\[
\Delta_x(d)=\one_S(x-d)+\one_S(x+d),
\]
because the new ordered pairs are $(a,x)$ and $(x,a)$ for $a\in S$. The greedy rule prioritizes candidates that cover the largest number of currently minimum-multiplicity offsets; ties are broken by coverage of currently zero offsets and then deterministically by the seeded ordering. A chunked vectorized implementation evaluates many candidates at once using Boolean indicator arrays.

\begin{algorithm}[H]
\caption{Vectorized greedy generator heuristic}
\label{alg:greedy}
\begin{algorithmic}[1]
\REQUIRE $n$, target degree $m$, random seed
\STATE $S\leftarrow\emptyset$, $\lambda(d)\leftarrow0$ for all $d$
\STATE $A\leftarrow\Zn\setminus\{0\}$
\WHILE{$|S|<m$}
    \STATE $M\leftarrow\min_{d\ne0}\lambda(d)$
    \STATE $D_M\leftarrow\{d\ne0:\lambda(d)=M\}$
    \STATE $D_0\leftarrow\{d\ne0:\lambda(d)=0\}$
    \FORALL{chunks $X\subseteq A$}
        \STATE evaluate $c_M(x)=|\{d\in D_M:\Delta_x(d)>0\}|$ for all $x\in X$
        \STATE evaluate $c_0(x)=|\{d\in D_0:\Delta_x(d)>0\}|$ for all $x\in X$
    \ENDFOR
    \STATE choose $x$ maximizing $(c_M(x),c_0(x))$ with seeded tie-breaking
    \STATE $S_{\mathrm{old}}\leftarrow S$
    \STATE $S\leftarrow S_{\mathrm{old}}\cup\{x\}$; $A\leftarrow A\setminus\{x\}$
    \STATE update $\lambda(x-a)$ and $\lambda(a-x)$ for every $a\in S_{\mathrm{old}}$
\ENDWHILE
\RETURN $S$
\end{algorithmic}
\end{algorithm}

\begin{proposition}[Greedy update formula]
For a candidate $x\notin S$ and a nonzero offset $d$, adding $x$ changes the multiplicity by
\[
\Delta_x(d)=\one_S(x-d)+\one_S(x+d).
\]
\end{proposition}

\begin{proof}
The only new ordered pairs involving $x$ are $(a,x)$ and $(x,a)$ for $a\in S$, plus $(x,x)$ which contributes only to offset zero. A pair $(a,x)$ contributes to nonzero offset $d$ exactly when $x-a=d$, equivalently $a=x-d\in S$. A pair $(x,a)$ contributes to $d$ exactly when $a-x=d$, equivalently $a=x+d\in S$. If $x-d\equiv0$ or $x+d\equiv0$, the corresponding indicator is zero under the standing network convention $0\notin S$; hence no self-loop generator is silently introduced. Summing the two indicators gives the formula.
\end{proof}

The vectorized greedy does not enumerate a witness proof of optimality. It is a construction heuristic whose output is validated exactly. Fig.~\ref{fig:greedy-conv} shows a typical convergence trace.

\begin{figure}[H]
\centering
\includegraphics[width=0.6\linewidth]{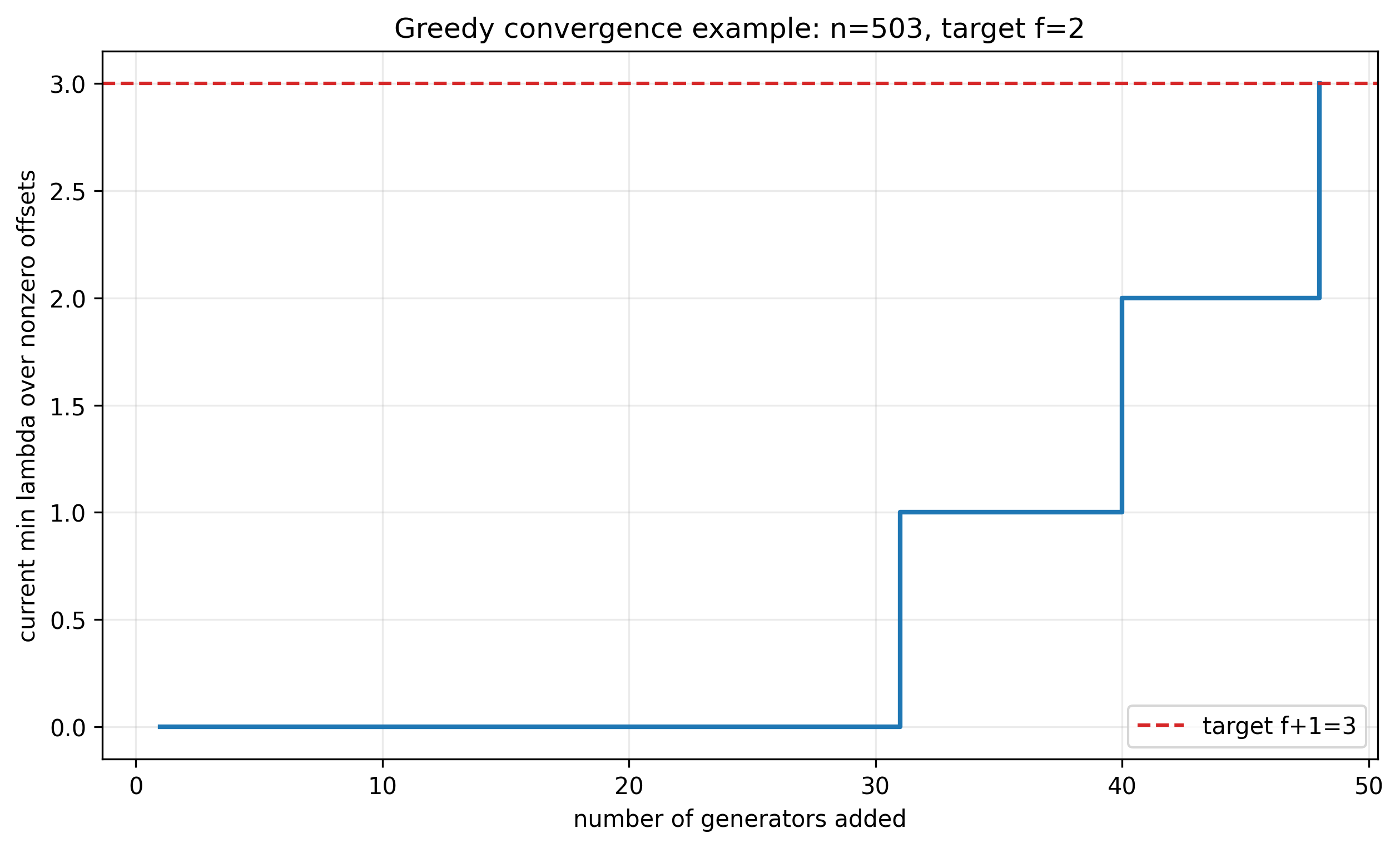}
\caption{Greedy convergence example. The curve shows the current minimum nonzero multiplicity as generators are added for $n=503$. The target $f+1=3$ is reached at the degree reported in Table~\ref{tab:thresholds}.}
\label{fig:greedy-conv}
\end{figure}

\section{Experimental Design}
\label{sec:experiments}

The optimized experiment evaluated $526{,}539$ circulant generator sets and wrote checkpointed CSV files after each batch. The tested sizes were
\[
n\in\{251,503,1009,2003,5003,10007\}
\]
and the fault levels were $f=0,1,2,3,4,5$. The tested families were interval, symmetric interval, modular stride, quadratic residues, random, and greedy max-min. Random sets were sampled with 1000 trials per tested degree. Greedy was run once per tested degree and was restricted to the threshold-relevant degree range. All reported values of $R(S)$, $\lambda$ averages, and standard deviations were recomputed exactly from the final generator sets. The optimized full run completed in $3995.020$ seconds and produced $1218$ recovery rows.

\subsection{Empirical Thresholds}

Table~\ref{tab:thresholds} reports the best degree found for each $(n,f)$. These values are empirical best-found degrees among the tested families, not exact values of $D_f(n)$.

\begin{table}[H]
\centering
\caption{Best found degrees for $f$-relay-fault tolerance. LB is the counting lower bound of Corollary~\ref{cor:degree-lb}. The ratio is $m/\!\sqrt{(f+1)n}$.}
\label{tab:thresholds}
\begin{adjustbox}{max width=\textwidth}
\begin{tabular}{rrrrrrr}
\toprule
$n$ & $f$ & LB & best $m$ & gap & $m/\sqrt{(f+1)n}$ & family \\
\midrule
251&0&17&20&3&1.262&greedy\\
251&1&23&28&5&1.250&greedy\\
251&2&28&33&5&1.203&greedy\\
251&3&33&38&5&1.199&greedy\\
251&4&36&41&5&1.157&greedy\\
251&5&40&45&5&1.160&greedy\\
503&0&23&31&8&1.382&greedy\\
503&1&33&40&7&1.261&greedy\\
503&2&40&48&8&1.236&greedy\\
503&3&46&54&8&1.204&greedy\\
503&4&51&59&8&1.176&greedy\\
503&5&56&65&9&1.183&greedy\\
1009&0&33&45&12&1.417&greedy\\
1009&1&46&59&13&1.313&greedy\\
1009&2&56&72&16&1.309&greedy\\
1009&3&64&78&14&1.228&greedy\\
1009&4&72&86&14&1.211&greedy\\
1009&5&79&92&13&1.182&greedy\\
2003&0&46&66&20&1.475&greedy\\
2003&1&64&86&22&1.359&greedy\\
2003&2&78&101&23&1.303&greedy\\
2003&3&90&114&24&1.274&greedy\\
2003&4&101&124&23&1.239&greedy\\
2003&5&111&135&24&1.231&greedy\\
5003&0&72&114&42&1.612&greedy\\
5003&1&101&142&41&1.420&greedy\\
5003&2&123&166&43&1.355&greedy\\
5003&3&142&184&42&1.301&greedy\\
5003&4&159&201&42&1.271&greedy\\
5003&5&174&217&43&1.252&greedy\\
10007&0&101&163&62&1.629&greedy\\
10007&1&142&209&67&1.477&greedy\\
10007&2&174&246&72&1.420&greedy\\
10007&3&201&271&70&1.355&greedy\\
10007&4&225&301&76&1.346&greedy\\
10007&5&246&319&73&1.302&greedy\\
\bottomrule
\end{tabular}
\end{adjustbox}
\end{table}

\begin{table}[H]
\centering
\caption{Exact $D_f(n)$ values and greedy calibration for small cyclic groups. The greedy column is the first degree at which the same seeded vectorized greedy pipeline reaches $R(S)\ge f+1$; the gap is greedy $m-D_f(n)$.}
\label{tab:exact-small}
\begin{adjustbox}{max width=\textwidth}
\begin{tabular}{rrrrrrrrrr}
\toprule
$n$ & $f$ & exact $D_f(n)$ & greedy $m$ & gap &
$n$ & $f$ & exact $D_f(n)$ & greedy $m$ & gap \\
\midrule
7&0&3&3&0 & 17&0&5&5&0\\
7&1&4&4&0 & 17&1&7&7&0\\
7&2&5&5&0 & 17&2&8&8&0\\
7&3&6&6&0 & 17&3&9&9&0\\
11&0&4&4&0 & 19&0&5&5&0\\
11&1&5&6&1 & 19&1&7&7&0\\
11&2&6&7&1 & 19&2&8&9&1\\
11&3&7&7&0 & 19&3&9&10&1\\
13&0&4&4&0 & 23&0&6&6&0\\
13&1&6&6&0 & 23&1&8&8&0\\
13&2&7&7&0 & 23&2&9&10&1\\
13&3&8&8&0 & 23&3&11&11&0\\
\bottomrule
\end{tabular}
\end{adjustbox}
\end{table}

\begin{table}[H]
\centering
\caption{Composite-size spot check using the same greedy construction. The ratios remain within the same practical range as the prime-size experiment, suggesting that the construction pipeline is not restricted to prime $n$.}
\label{tab:composite}
\begin{tabular}{rrrrr}
\toprule
$n$ & $f$ & LB & greedy $m$ & ratio \\
\midrule
256&0&17&21&1.313\\
256&2&29&34&1.227\\
256&5&40&45&1.148\\
512&0&24&31&1.370\\
512&2&40&48&1.225\\
512&5&56&65&1.173\\
1024&0&33&46&1.438\\
1024&2&56&70&1.263\\
1024&5&79&93&1.186\\
\bottomrule
\end{tabular}
\end{table}

The ratios are not monotone in $n$ for fixed $f$, and for $f=0$ they grow from $1.26$ to $1.63$ across the tested range. This should not be interpreted as a theorem about $D_f(n)$; it is a property of the tested heuristic and degree grid.

\begin{figure}[H]
\centering
\includegraphics[width=0.6\linewidth]{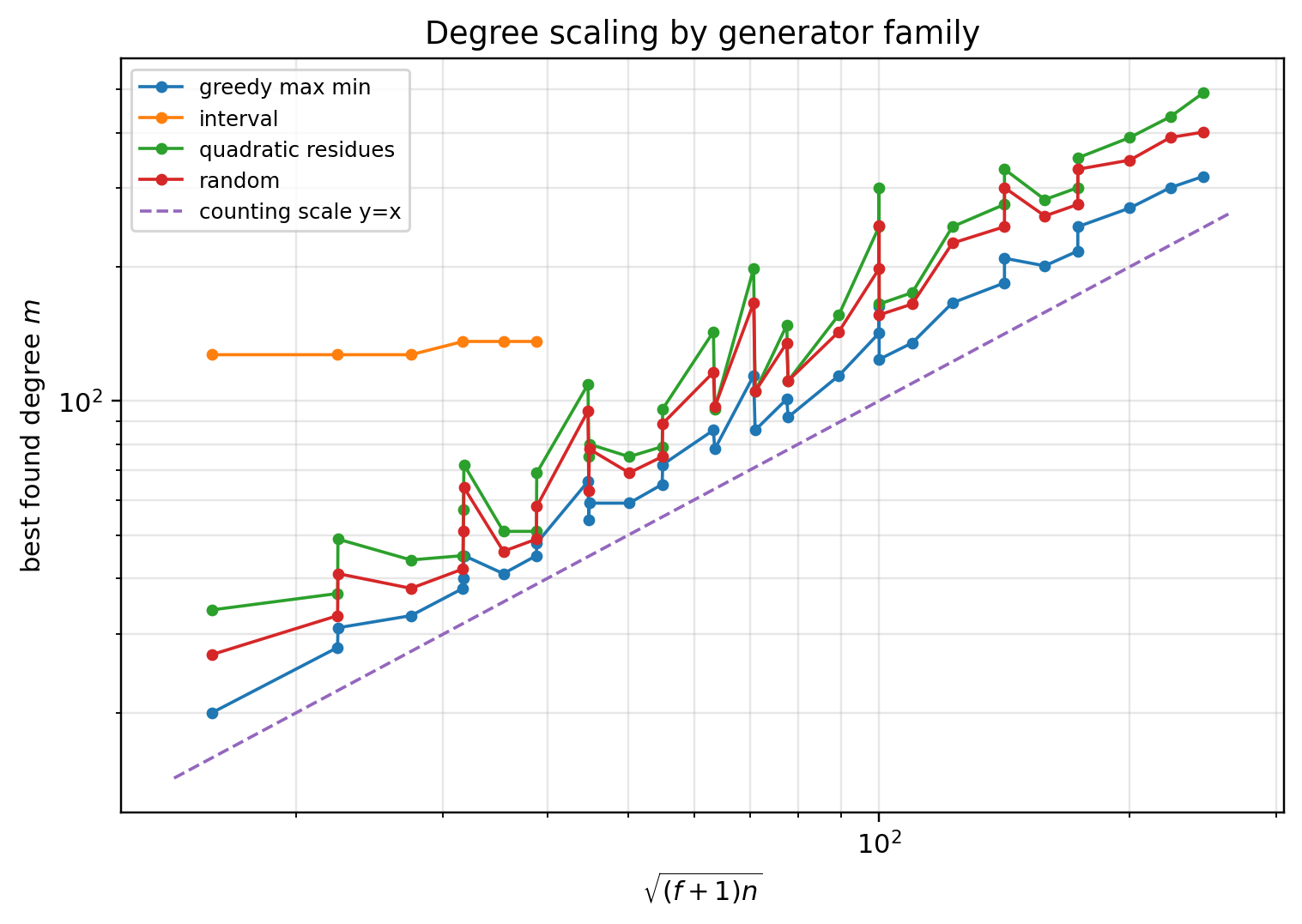}
\caption{Degree scaling against the natural lower-bound scale $\sqrt{(f+1)n}$. Greedy sets track the lower-bound scale more closely than interval sets, while quadratic-residue designs sometimes outperform greedy at larger degrees. The horizontal axis is shown on a logarithmic scale to keep the weak interval-family points visible. The diagonal is not a theorem of achievability; it is the counting scale.}
\label{fig:scaling}
\end{figure}

\subsection{High-Redundancy Distributions}

Table~\ref{tab:high} reports representative high-redundancy designs. Average multiplicity can be much larger than the worst-case multiplicity $R(S)$; both are relevant. The worst case gives adversarial guarantees, while the distribution shapes explain random-failure behavior.

\begin{table}[H]
\centering
\caption{Representative high-redundancy designs from the full run.}
\label{tab:high}
\begin{tabular}{rrrrr}
\toprule
$n$ & $m$ & family & $R(S)$ & mean $\lambda$ \\
\midrule
251 & 156 & random & 92 & 96.72\\
503 & 220 & quadratic residues & 89 & 95.98\\
1009 & 312 & quadratic residues & 83 & 96.26\\
2003 & 439 & random & 77 & 96.04\\
5003 & 694 & random & 73 & 96.15\\
10007 & 981 & random & 71 & 96.08\\
\bottomrule
\end{tabular}
\end{table}

\begin{figure}[H]
\centering
\includegraphics[width=0.6\linewidth]{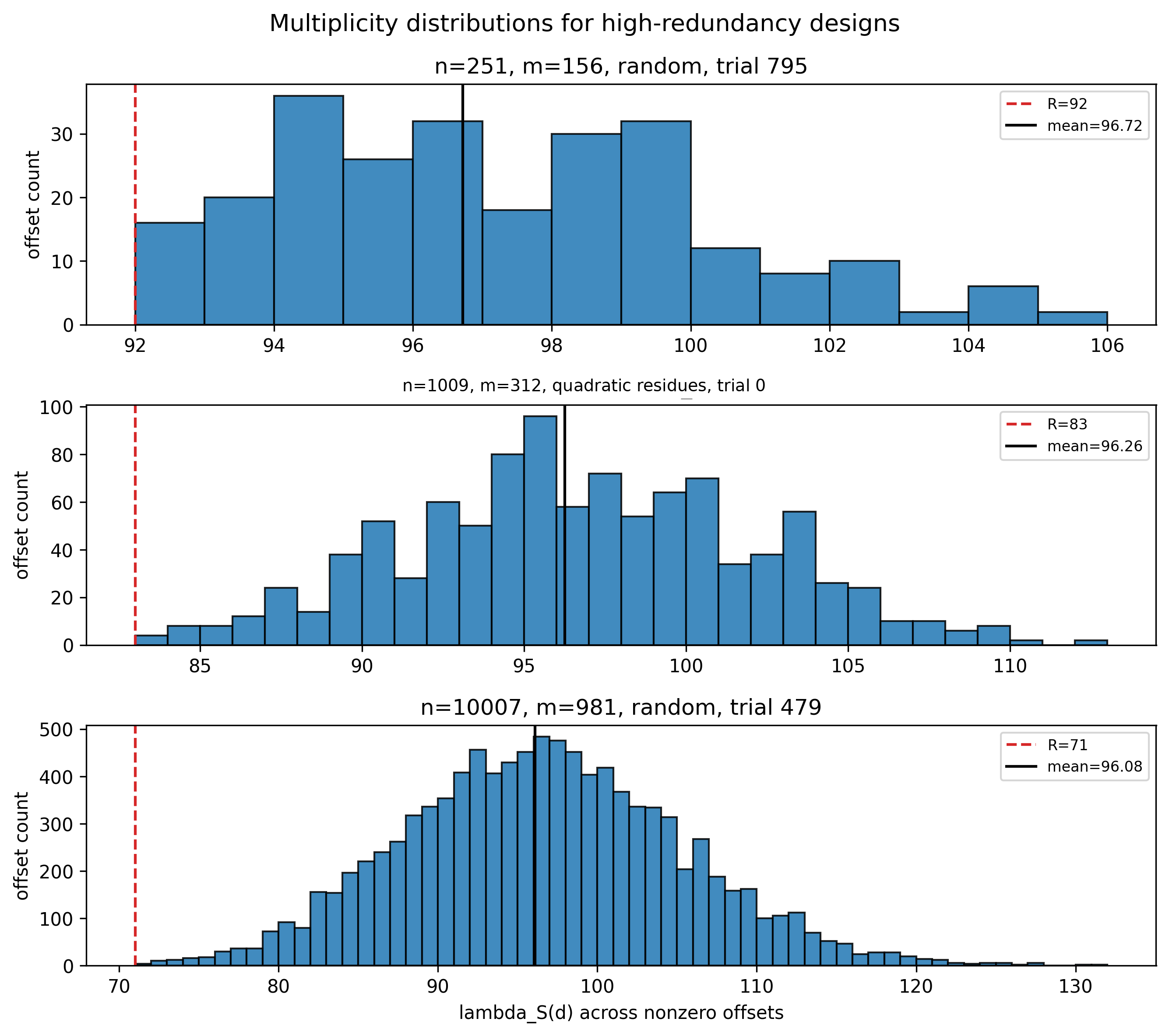}
\caption{Multiplicity distributions for high-redundancy designs. The gap between the mean and $R(S)$ shows why average redundancy and worst-case redundancy must be reported separately. Each panel uses its own horizontal scale to show distribution shape.}
\label{fig:hist}
\end{figure}

\subsection{Recovery Simulation}

For the selected best designs, random relay failures were simulated over $5000$ trials per row. Table~\ref{tab:recovery} aggregates the survival probabilities by number of failed relays. The minimum is the worst observed simulation row, not a theoretical lower bound.

\begin{table}[H]
\centering
\caption{Random relay-failure recovery over 1218 simulation rows. The minimum row for every $q$ came from the smallest threshold design $n=251$, $m=20$, greedy, $R=1$; thus the low $q=20$ minimum is a small-degree outlier rather than a systematic failure of high-redundancy designs.}
\label{tab:recovery}
\begin{tabular}{rrrr}
\toprule
failures $q$ & min success & mean success & median \\
\midrule
1 & 0.9984 & 0.999936 & 1.0000\\
2 & 0.9942 & 0.999848 & 1.0000\\
3 & 0.9946 & 0.999803 & 1.0000\\
4 & 0.9908 & 0.999723 & 1.0000\\
5 & 0.9892 & 0.999682 & 1.0000\\
10 & 0.9792 & 0.999338 & 1.0000\\
20 & 0.9496 & 0.998581 & 1.0000\\
\bottomrule
\end{tabular}
\end{table}

\subsection{Greedy Runtime and Construction Cost}

Table~\ref{tab:greedy-runtime} gives empirical worker times for the vectorized greedy heuristic. The table also explains why greedy construction should be regarded as an offline design step rather than an online routing step. Online recovery uses only the precomputed tables of Section~\ref{sec:algorithms}.

\begin{table}[H]
\centering
\caption{Empirical vectorized-greedy construction cost. Worker time is summed over greedy designs for that $n$ and is not wall-clock time.}
\label{tab:greedy-runtime}
\begin{tabular}{rrrrr}
\toprule
$n$ & designs & max $m$ & max $R$ & max sec/design \\
\midrule
251 & 53 & 90 & 78 & 0.09\\
503 & 59 & 126 & 72 & 0.47\\
1009 & 75 & 176 & 86 & 2.45\\
2003 & 80 & 247 & 81 & 17.29\\
5003 & 87 & 390 & 78 & 182.48\\
10007 & 89 & 552 & 73 & 1042.51\\
\bottomrule
\end{tabular}
\end{table}

\begin{table}[H]
\centering
\caption{Seed sensitivity of the vectorized greedy heuristic on two threshold cases.}
\label{tab:seed-sensitivity}
\begin{tabular}{rrrrrr}
\toprule
$n$ & $m$ & target $f$ & seed & $R(S)$ & std. dev. \\
\midrule
251&33&2&0&3&1.179\\
251&33&2&1&3&1.080\\
251&33&2&2&3&1.245\\
251&33&2&3&3&1.332\\
251&33&2&4&3&1.270\\
1009&72&2&0&3&1.498\\
1009&72&2&1&3&1.444\\
1009&72&2&2&3&1.446\\
1009&72&2&3&3&1.550\\
1009&72&2&4&3&1.459\\
\bottomrule
\end{tabular}
\end{table}

\subsection{Comparison With Difference-Basis Benchmarks}

The counting bound is an unavoidable lower bound. Existing cyclic difference-basis theory gives general existence upper bounds for $f=0$, such as the $1.5\sqrt n$-scale bound in \cite{BanakhGavrylkiv2019}. For higher multiplicity, recent preprints study $g$-difference bases, but we treat those results as preprint-level context rather than as peer-reviewed cyclic construction constants. Table~\ref{tab:db-comparison} therefore separates three quantities: the counting lower bound, an external peer-reviewed $f=0$ cyclic difference-basis upper bound when applicable, and the empirical relay-design threshold found by the greedy pipeline.

\begin{table}[H]
\centering
\caption{Benchmark comparison. BG is the $f=0$ cyclic difference-basis construction upper bound from \cite{BanakhGavrylkiv2019}. For $f=1,2$, no peer-reviewed explicit cyclic $g$-difference-basis generator bound is used here; the benchmark column therefore reports the counting lower bound as the concrete reference point.}
\label{tab:db-comparison}
\begin{adjustbox}{max width=\textwidth}
\begin{tabular}{rrrrrrl}
\toprule
$n$ & $f$ & counting LB & external benchmark & greedy $m$ & ratio & interpretation \\
\midrule
251&0&17&BG $\le24$&20&1.26&greedy within BG scale\\
251&1&23&LB $23$&28&1.25&--\\
251&2&28&LB $28$&33&1.20&--\\
1009&0&33&BG $\le48$&45&1.42&greedy within BG scale\\
1009&1&46&LB $46$&59&1.313&--\\
1009&2&56&LB $56$&72&1.309&--\\
10007&0&101&BG $\le151$&163&1.63&algebraic benchmark smaller\\
10007&1&142&LB $142$&209&1.48&--\\
10007&2&174&LB $174$&246&1.42&--\\
\bottomrule
\end{tabular}
\end{adjustbox}
\end{table}

\section{Practical Interpretation for Interconnection Networks}
\label{sec:practical}

The shared-relay model is most appropriate when a system benefits from immediate two-hop alternatives rather than global rerouting. Examples include local rendezvous, acknowledgement aggregation, duplicated control signaling, and small-message recovery in symmetric overlays or structured on-chip fabrics. Ring interconnects in commercial processors are the degree-two circulant baseline; recent NoC work also implements and evaluates routing algorithms for circulant topologies \cite{RomanovAccess2020,MonakhovaTNSE2023}.

A concrete deployment is a duplicated control-notification primitive on a directed circulant overlay. When a controller or source selects a terminal pair $(u,v)$ that must both receive a short control word, it computes $d=v-u$ and scans $P_d$ to choose a surviving relay $r$. Because the model requires $r\to u$ and $r\to v$, the selected relay can forward the duplicated control word to both terminals. If $r$ is failed, the failed-node bitmap causes the endpoints or controller to skip the same table entry and select the next surviving relay.

By Corollary~\ref{cor:preprocess}, the preprocessing cost is $O(m^2)$. The memory cost is governed by Proposition~\ref{prop:table-size}. At the largest threshold design in Table~\ref{tab:thresholds}, $m=319$ and hence $m(m-1)=101{,}442$ offset-pair entries. Proposition~\ref{prop:lookup-cost} makes this latency distinction explicit: relay lookup is $O(\lambda_S(d))\le O(m)$ in the worst case and $m(m-1)/(n-1)$ on average under a uniform-offset traffic model, whereas a graph-search rerouting baseline on a degree-$m$ circulant may inspect $\Theta(nm)$ edges. Table~\ref{tab:latency} gives a concrete software measurement of this separation for four threshold designs.

\begin{table}[H]
\centering
\caption{Software recovery microbenchmark for greedy $f=2$ threshold designs with two random failed relays and $10{,}000$ random terminal pairs. The proposed lookup uses Algorithm~\ref{alg:lookup}; the graph-search baseline scans candidate relays in the circulant graph from scratch. Timings are Python wall-clock microseconds.}
\label{tab:latency}
\begin{tabular}{rrrrrr}
\toprule
$n$ & $m$ & lookup mean & lookup 99th & search mean & search 99th \\
\midrule
251&33&0.245&0.582&3.515&13.367\\
503&48&0.299&0.691&7.301&27.502\\
1009&72&0.297&0.516&14.309&58.344\\
2003&101&0.471&2.197&25.625&100.648\\
\bottomrule
\end{tabular}
\end{table}

The shared-relay primitive is most useful when $m$ is small relative to $n$, so that $P_d$ is sparse and the $O(m)$ lookup cost is substantially smaller than a graph-wide search; when $m=\Theta(n)$, most offsets are covered trivially and conventional routing may be sufficient. The limitation is that the model guarantees only the existence of a two-hop relay with outgoing links to both terminals. It does not by itself solve congestion control, global packet scheduling, or multi-hop routing after arbitrary topology damage. These are deliberately separated: the paper defines and optimizes the relay primitive, which can be embedded into larger routing protocols.

\section{Discussion}
\label{sec:discussion}

Shared-relay fault tolerance is equivalent to cyclic difference multiplicity, a known combinatorial condition. The equivalence becomes practically useful only after adding a network-layer interpretation: relay tables, failure semantics, selection rules, memory costs, balance guarantees, and empirical generator design. This separation protects the paper from rediscovering difference bases while giving the equivalence operational meaning for circulant interconnection networks.

The empirical results show three trends. First, simple interval-like generator sets are poor worst-case shared-relay designs because they leave many offsets uncovered. Second, vectorized greedy sets can reach $R(S)=f+1$ within a moderate constant factor of the counting lower bound for all tested $n$ and $f\le5$. Third, random failure performance is much stronger than adversarial guarantees, with mean survival above $0.998$ even at $20$ random failures and above $0.999$ for $q\le10$, but the adversarial guarantee remains the correct safety definition for worst-case relay failures.

Several limitations remain. First, the greedy construction is a heuristic and should not be interpreted as an exact computation of $D_f(n)$ outside the small cases exhaustively verified in Table~\ref{tab:exact-small}. Second, we prove membership of the design decision problem in NP, but we do not claim an NP-hardness theorem or a worst-case approximation ratio for Algorithm~\ref{alg:greedy}. Third, for $f\ge1$, no peer-reviewed explicit cyclic $g$-difference-basis construction is used here to provide a concrete numerical comparison. Fourth, the current systems evidence is a software-level recovery microbenchmark, not a full packet-level congestion or hardware-cycle simulation.

\section{Conclusion}

This paper developed a fault-tolerant shared-relay communication framework for directed circulant interconnection networks. The core equivalence to cyclic difference multiplicity was stated explicitly and used as a tool rather than claimed as a new combinatorial object. On top of that equivalence, the paper defined network design parameters $D_f(n)$ and $R(n,m)$, proved worst-case and random-failure recovery guarantees, proved a structural failure theorem for interval circulants, developed relay lookup and load-aware selection algorithms, quantified table size and preprocessing, calibrated exact $D_f(n)$ values for small groups, measured relay lookup against a graph-search baseline, gave worked examples, and evaluated more than half a million generator sets. The results support the view that shared-relay multiplicity is a useful local robustness metric for circulant networks and that generator choice strongly affects worst-case relay-fault tolerance.

\appendix
\section{Reproducibility Notes}

The optimized experiment uses checkpointed CSV output. Layer 1 writes exact metrics and generator sets after each completed batch. Layer 2 reads only the selected compact designs for recovery simulation. The reported greedy construction is a heuristic, but every reported $R(S)$ and multiplicity statistic is recomputed exactly from the final generator set. The artifact files include the generator sets, degree summaries, recovery summaries, selected histograms, run settings, and figure scripts. The optimized full run reported here was executed on a Windows workstation with 56 logical CPU workers available; the production script capped design workers at 24 and recovery workers at 8, using Python 3.x and NumPy for vectorized difference counting.

\section*{Acknowledgment}
The authors thank the Department of Computer Science, Faculty of Science, Kuwait University, for its support. The authors also thank the reviewers and colleagues whose comments helped clarify the distinction between the underlying difference-basis theory and the network-layer contribution of this work. This research received no specific grant from any funding agency in the public, commercial, or not-for-profit sectors.

\end{document}